\documentclass[aps,11pt,twoside, nofootinbib]{revtex4}
\usepackage{amsmath,latexsym,amssymb,verbatim,enumerate,graphicx}
\usepackage{hyperref}
\usepackage{color}

\usepackage{theorem}
\newtheorem{definition}{Definition}[section]
\newtheorem{proposition}[definition]{Proposition}
\newtheorem{lemma}[definition]{Lemma}

\newtheorem{theorem}{Theorem}

\newcommand{\ot}{\otimes}
\newcommand{\be}{\begin{equation}}
\newcommand{\ee}{\end{equation}}
\def\hcal{{\cal H}}

\def\BQP{{\sf{BQP}}}
\def\PH{{\sf{PH}}}

\def\FC{\textsc{Fourier Checking}}
\def\CC{\textsc{Circuit Checking}}
\def\BPP{{\sf{BPP}}}

\def\SZK{{\sf{SZK}}}

\def\NP{{\sf{NP}}}

\def\PSPACE{{\sf{PSPACE}}}
\def\P{{\sf{P}}}
\def\U{{\sf{U}}}

\def\RFS{{\sf{RFS }}}

\def\squareforqed{\hbox{\rlap{$\sqcap$}$\sqcup$}}
\def\qed{\ifmmode\squareforqed\else{\unskip\nobreak\hfil
\penalty50\hskip1em\null\nobreak\hfil\squareforqed
\parfillskip=0pt\finalhyphendemerits=0\endgraf}\fi}
\def\endenv{\ifmmode\;\else{\unskip\nobreak\hfil
\penalty50\hskip1em\null\nobreak\hfil\;
\parfillskip=0pt\finalhyphendemerits=0\endgraf}\fi}
\newenvironment{proof}[1][Proof]{\noindent \textbf{{#1~} }}{\qed}

\newcommand{\bra}[1]{\langle #1|}
\newcommand{\ket}[1]{|#1\rangle}
\newcommand{\braket}[2]{\langle #1|#2\rangle}
\newcommand{\tr}{\text{tr}}

\newcommand{\id}{\mathbb{I}}

\mathchardef\ordinarycolon\mathcode`\:
\mathcode`\:=\string"8000
\def\vcentcolon{\mathrel{\mathop\ordinarycolon}}
\begingroup \catcode`\:=\active
  \lowercase{\endgroup
  \let :\vcentcolon
  }

\newcommand{\nc}{\newcommand}
\nc{\rnc}{\renewcommand} \nc{\beq}{\begin{equation}}
\nc{\eeq}{{\end{equation}}} \nc{\bea}{\begin{eqnarray}}
\nc{\eea}{\end{eqnarray}} \nc{\beqa}{\begin{eqnarray}}
\nc{\eeqa}{\end{eqnarray}} \nc{\lbar}[1]{\overline{#1}}

\nc{\conv}{\operatorname{conv}}
\nc{\smfrac}[2]{\mbox{$\frac{#1}{#2}$}} \nc{\Tr}{\operatorname{Tr}}
\nc{\ox}{\otimes} \nc{\dg}{\dagger} \nc{\dn}{\downarrow}
\nc{\lmax}{\lambda_{\text{max}}}
\nc{\lmin}{\lambda_{\text{min}}}

\nc{\csupp}{{\operatorname{csupp}}}
\nc{\qsupp}{{\operatorname{qsupp}}} \nc{\var}{\operatorname{var}}
\nc{\rar}{\rightarrow} \nc{\lrar}{\longrightarrow}
\nc{\poly}{\operatorname{poly}}
\nc{\polylog}{\operatorname{polylog}} \nc{\Lip}{\operatorname{Lip}}
\nc{\mb}[1]{\mathbf{#1}}
\nc{\ep}{\epsilon}
\nc{\Om}{\Omega}
\nc{\wt}[1]{\widetilde{#1}}

\def\>{\rangle}
\def\<{\langle}

\nc{\glneq}{{\raisebox{0.6ex}{$>$}  \hspace*{-1.8ex} \raisebox{-0.6ex}{$<$}}}
\nc{\gleq}{{\raisebox{0.6ex}{$\geq$}\hspace*{-1.8ex} \raisebox{-0.6ex}{$\leq$}}}

\nc{\vholder}[1]{\rule{0pt}{#1}}
\nc{\wh}[1]{\widehat{#1}}
\nc{\h}[1]{\widehat{#1}}

\nc{\ob}[1]{#1}

\def\beq{\begin {equation}}
\def\eeq{\end {equation}}

\def\be{\begin{equation}}
\def\ee{\end{equation}}

\nc{\eq}[1]{Eq.~(\ref{eq:#1})} \nc{\eqs}[2]{Eqs.~(\ref{eq:#1}) and
(\ref{eq:#2})}

\nc{\eqn}[1]{Eq.~(\ref{eqn:#1})}
\nc{\eqns}[2]{Eqs.~(\ref{eqn:#1}) and (\ref{eqn:#2})}

\nc{\region}{\cS\cW}

\begin{document}

\title{{\Large Exponential Quantum Speed-ups are Generic}}

\author{Fernando G.S.L. Brand\~ao}
\email{fgslbrandao@gmail.com}
\affiliation{Departamento de F\'isica, Universidade Federal de Minas Gerais,
     Belo Horizonte, Caixa Postal 702, 30123-970, MG, Brazil}
\author{Micha\l{} Horodecki}
 \email{fizmh@ug.edu.pl}
\affiliation{Institute for Theoretical Physics and Astrophysics, University of Gda\'nsk, 80-952 Gda\'nsk, Poland}


\begin{abstract}

A central problem in quantum computation is to understand which quantum circuits are useful for exponential speed-ups over classical computation. 
We address this question in the setting of query complexity and show that for almost any sufficiently long quantum circuit one can construct a black-box problem 
which is solved by the circuit with a constant number of quantum queries, but which requires exponentially many classical queries, even if the classical machine has the ability to postselect.  

We prove the result in two steps. In the first, we show that almost any element of an approximate unitary 3-design is useful to solve a certain black-box problem efficiently. The problem is based on a recent oracle construction of Aaronson and gives
an exponential separation between quantum and classical post-selected bounded-error query complexities. 

In the second step, which may be of independent interest, we prove that linear-sized random quantum circuits give an approximate unitary 3-design. The 
key ingredient in the proof is a technique from quantum many-body theory to lower bound the spectral gap of local quantum Hamiltonians. 

\end{abstract}

\maketitle

\parskip .75ex


\section{Introduction}


Quantum computation holds the promise of solving certain problems substantially faster than classical computation. The most famous example is arguably Shor's polynomial-time quantum algorithm for factoring \cite{Sho97}, a task which is believed to require exponential time in a classical computer. Other problems for which quantum algorithms appear to be give exponential speed-ups include simulating quantum systems \cite{Fey82}, solving Pell's equation \cite{Hal02}, approximating the Jones polynomial \cite{FKW02, AJL06}, and estimating certain properties of sparse systems of linear equations \cite{HHL09}.  Unfortunately, the apparent computational superiority of quantum mechanics is presently only conjectural. In fact, one cannot hope to separate the class of problems solved in polynomial time by quantum and classical computation without settling major open questions in computational complexity theory\footnote{such as $\P  \stackrel{?}{=} \PSPACE$.}.      

A setting for which quantum computation is \textit{provably} superior to classical  is the one of \textit{query complexity} (also known as decision tree complexity or black-box complexity). There one is given the ability to query a black-box function  and the goal is to determine a certain property of the function. The complexity of the problem  is measured by the minimum number of queries needed to determine such property. In the quantum case, one is able to query the black-box in superposition, a feature which potentially renders it more powerful than the classical one. 

The first example of a black-box problem exhibiting a superpolynomial separation of  quantum and randomized classical query complexities was the recursive fourier sampling ($\RFS$) problem of Bernstein and Vazirani \cite{BV97}. Soon after it, Simon presented a black-box problem with an \textit{exponential} quantum-classical separation \cite{Sim97}; Simon's problem is also a good example of the usefulness of the query complexity model for the development of new  algorithms: its quantum solution was both a motivation for and an important element in Shor's quantum algorithm for factoring  \cite{Sho97}. Many other oracle separations have since been found, see e.g. \cite{Kit95, ME99, dBCW02, vDH00, Wat01, CCD+03}. In terms of complexity classes, these query complexity results show the existence of an oracle $U$ for which  $\BQP^U \neq \BPP^U$.\footnote{see the complexity zoo (http://qwiki.stanford.edu/wiki/Complexity$\_$Zoo) for definitions of the standard complexity classes.}

Having collected  evidence that quantum computation is superior to randomized classical computation, it is interesting to get insight about where exactly does $\BQP$ sit in the zoo of classical complexity classes. For example, are there problems that a quantum computer can solve efficiently, but which a classical computer cannot even check a potential solution in reasonable time? This is the question whether $\BQP  \subseteq \NP$  and already in the seminal paper \cite{BV97}, the $\RFS$ problem was used to build an oracle $U$ such that $\BQP^U  \nsubseteq \NP^U$. One can go even further and ask for an oracle for which $\BQP$ is not contained in the entire polynomial hierarchy (\PH). In \cite{BV97} it was conjectured that the $\RFS$ problem also gives an oracle relative to which $\BQP  \nsubseteq \PH$, but whether this is indeed the case remains an open question.

Recently, Aaronson  \cite{Aar09}  proposed an interesting new oracle problem as a candidate to put $\BQP$ outside $\PH$.\footnote{In \cite{Aar09} it was shown that the separation would follow from a certain generalization of the Linial-Nisan conjecture \cite{LN90} recently settled by Braverman \cite{Bra09}. However this generalization was later falsified in \cite{Aar10}.} Although the usefulness of this oracle for the $\BQP$ vs. $\PH$ question still has to be elucidated, the problem was shown to have a huge separation of quantum and classical query complexities: it can be solved by a constant number of quantum queries, while it requires exponentially many queries by a classical machine, even if we give the classical machine the -- extremely powerful -- ability to \textit{postselect} on a given result of  the computation. This is the strongest separation of quantum and classical query complexities to date. It also implies oracles relative to which $\BQP \nsubseteq \BPP_{\text{path}}$ \footnote{Here $\BPP_{\text{path}}$ is defined as the class of problems which can be solved in polynomial time, with high probability, by a randomized classical computer which can postselect on given outcomes of the computation \cite{HHT97}.} and  $\BQP \nsubseteq \SZK$, which supersedes all previous oracle separations for $\BQP$.

Aaronson's problem, named $\FC$, is the following: We are given two boolean functions $f, g : \{ 0, 1 \}^n \rightarrow \{ -1, 1 \}$ with the promise that either 
\begin{itemize}
\item $f$ and $g$ are chosen uniformly at random, or 
\item for a vector $v \in \mathbb{R}^{2^n}$ with entries $v_x$ drawn independently from a normal distribution of mean 0 and variance 1, the functions are chosen as $f(x) = \text{sgn}(v_x)$ and $g(x) = \text{sgn}(\hat{v}_x)$ \footnote{The sign function is defined as: $\text{sgn}(x) = 1$ for $x \geq 0$, and $\text{sgn}(x) = -1$ otherwise.}. Here the vector $\hat{v}$ is the Fourier transform over $\mathbb{Z}_2^n$ of $v$ and is given 
by  
\begin{equation} \label{fouriertransform}
\hat{v}_x = \sum_{y \in \{0, 1\}^n} (-1)^{x.y}v_y. 
\end{equation} 
\end{itemize}
The task is to decide which is the case.
In words, we should determine if the two functions are not correlated at all or if one of them is well correlated with the Fourier transform of the other. 

The quantum algorithm proposed in \cite{Aar09} to solve the problem is particularly simple. One prepares the  uniform superposition over the computational basis, queries $f$, applies the quantum Fourier transform (QFT), queries $g$, and checks if the final state is again in a uniform superposition over the computational basis. If the functions are independent, then there is only an exponential small chance of getting the right outcome in the final measurement, while in  the case where they are correlated, this happens with constant probability.

Considering how well this problem fleshes out the superiority of quantum computation to classical, it is worthwhile to try to understand what exactly  gives its strength. For instance, what is the role played by the Fourier transform, both the the definition of the problem and in the quantum algorithm solving it? Can we replace it by some other transformation? One of the goals of this paper is to shed  light on these questions.       
  
From a broader perspective, we will be concerned with the following question, central to our understanding of the computational capabilities offered by quantum mechanics: What is  the set of quantum circuits which provide large quantum speed-ups? More precisely, for which quantum circuits can we construct black-box problems which are solved by the circuit with only a few queries to the black-box, but which require a large number of queries  for randomized classical computation? This question is in  a sense a converse to the well-studied problem  of characterizing the class of black-box functions allowing for significant quantum speed-ups (see e.g. \cite{BBC+01, AA09}). While the latter  deals with the determination of which computational problems are suited for quantum computing, the former contributes to the classification of which quantum algorithmic techniques are  useful for solving problems efficiently.   

For instance, all the early examples of quantum algorithms offering superpolynomial speed-ups \cite{Deu85, DJ92, BV97, Sim97} were based on the quantum Fourier transform and this led to the speculation that it could be the defining aspect of quantum computation behind quantum speed-ups. Subsequently, other black-box problems showing a quantum advantage were found having no relation to the QFT \cite{CCD+03, AJL06, HH08}, hence extending the scope of techniques for constructing quantum algorithms.

Of particular note in this context, and for this paper, is the work of Hallgren and Harrow \cite{HH08} on generalizations of Bernstein and Vazirani's $\RFS$ problem. The $\RFS$ classical--quantum separation is built in two steps: first one construct a black-box problem requiring a constant number of quantum queries, but $\Omega(n)$ classical queries. Then one uses recursion to boost the separation to a $n^{O(1)}$ quantum versus $n^{\log(n)}$ classical queries. The oracle problem in the first part is based on the Fourier transform and solved by the QFT. In \cite{HH08} it was shown that this problem could be modified to have almost \textit{any} quantum circuit (from a natural measure on circuits) in the place of the Fourier transform and still achieve the constant versus linear separation, as in the original formulation. Moreover, any such problem could also be boosted by recursion to provide a black-box problem with a superpolynomial quantum-to-classical gap in query complexity.    
              
\subsection{Our results}

In this paper we generalize Aaroson's $\FC$ problem \cite{Aar09} and show that the Fourier transform, both in the definition of the problem and in the quantum algorithm solving it, can be replaced by a large class of quantum circuits. These include both the Fourier transform over any (possibly non-abelian) finite group and \textit{almost any} sufficiently long quantum circuit from a natural distribution on the set of quantum circuits, which we discuss later on. We obtain exponential separations of quantum and postselected classical query complexities for all such circuits. 

\vspace{0.1 cm}
\noindent \textbf{Flat circuits imply exponential separation:} In more detail, we first introduce a simple measure of flatness, or dispersiveness, of a unitary $U$ on $n$ qubits, denoted $C(U)$ . It is defined as the minimal \textit{min-entropy}\footnote{For a probability distribution $p(x)$ we define its min-entropy as $h_{\text{min}}(p) := - \log  \max_{x} p(x)$} (over $j \in \{ 0, .., 2^{n}-1 \}$)  of the outcome probability distribution of a computational basis measurement applied to $U\ket{j}$. For $N := 2^{n}$,
\begin{equation}
C(U) :=  \min_{j \in [N]} h_{\text{min}} \left( \{ |\bra{0}U\ket{j}|^{2}, ..., |\bra{N-1}U\ket{j}|^{2} \}  \right),
\end{equation} 
with $[N] := \{ 0, ..., N - 1 \}$. It thus measures the worst-case dispersiveness of states obtained by applying $U$ to computational basis states. 

In section \ref{oracleproblem} we define, for a unitary $U$, the black-box problem U-$\CC$, a variant of $\FC$ in which the Fourier transform in the definition of the vector $\hat{v}$ (given by Eq. (\ref{fouriertransform})) is replaced by $U$. The problem is constructed so that a quantum computer can easily solve it with access to a few realizations of the unitary $U$, while it is classically hard for any $U$ with large $C(U)$.

In detail, on one hand we prove a lower bound of $2^{\Omega(C(U))}$ on the classical query complexity with postselection of U-$\CC$ (see section \ref{classicallowerbounds}). Following the ideas of \cite{Aar09}, we do so by showing that the discretized version of the random vector $(v, Uv)$ -- for a vector $v$ composed of independent elements $v_x$ each drawn from a normal distribution of mean 0 and variance 1 -- is $k^{O(1)}2^{\Omega(C(U))}$-almost $k$-wise independent (a fact which was shown to imply the previous exponential lower bound on the postselected classical query complexity \cite{Aar09}). 

On the other hand, on a quantum computer we can solve U-$\CC$ by the following simple modification of Aaaroson's algorithm: we prepare each qubit in the $\ket{+} := (\ket{0} + \ket{1})/\sqrt{2}$ state, forming the uniform superposition over the computational basis. Then we query the $f$ function, apply the circuit $U$, query the $g$ function, and measure each qubit in the Hadamard basis, accepting if all of them are found in the $\ket{+}$ state. Therefore we obtain:

\begin{theorem}  \label{maindesignals} 
For any circuit $U$ acting on $n$ qubits with $C(U) = \Omega(n)$, the problem $\U$-$\CC$ shows an exponential separation of quantum and postselected classical query complexities. 
\end{theorem}


We then proceed by giving two classes of unitaries with $C(U) = \Omega(n)$. 

\begin{theorem} \label{classeswithlargeC}
\begin{list}{\quad}{}
\item[]
\item[] (i) Let $U_{\text{QFT}}(G)$ be the quantum Fourier transform over the finite group $G$. Then $C(U_{\text{QFT}}(G)) \geq \frac{1}{2}\log |G|$. 
\item[] (ii)  Given any $2^{-3tn}$-approximate unitary \textit{t}-design on $n$ qubits, all but a $2^{-(t(1 - \beta) - 2)n + 1}$ fraction of its elements have $C(U) \geq \beta n$.
\end{list}
\end{theorem}

In particular, we find that for $2^{-9n}$-approximate unitary 3-designs, all but a $2^{-n/2 + 1}$ fraction of its element have $C(U) \geq n / 6$. We note that the result of the theorem does not appear to hold true for unitary 2-designs and thus we seem to have the first application of unitary $t$-designs for $t > 2$.

The proofs of both statements of Theorem \ref{classeswithlargeC} are elementary and are given in section \ref{familiesflat}. 

\vspace{0.1 cm}
\noindent \textbf{Random circuits are unitary 3-designs:} A unitary $t$-design is an ensemble of unitaries $\{ \mu(dU), U\}$, for a measure $\mu$ on the set of unitaries, such that the average (over $\mu$) of any $t$-degree polynomial on the entries of $U$ and their complex conjugates is equal to the average over the Haar measure. An approximate unitary $t$-design is a relaxed version of the previous definition, in which we only require that the averages are close to each other (see section \ref{approximateunitarydesigns} for a precise definition) \cite{DCEL09, GAE07}. 

In a series of papers \cite{ELL05, ODP07, DOP07, HL09, DJ10} it was established that polynomially long random quantum circuits constitute an approximate unitary 2-design. The random quantum circuit model used is the following: in each step a random pair of qubits is chosen and a gate from a universal set of gates, also chosen at random, is applied to them. Although there is evidence that random quantum circuits of polynomial lenght are unitary $t$-design for every $t = \text{poly}(n)$ \cite{AB08, BV10}, this has not been rigorously proved so far, even for the $3$-design case. 

Here we prove that random quantum circuits are indeed approximate unitary $3$-designs. We show it both for the random circuit model of the previous paragraph and for a different one, introduced in \cite{HP07} as a toy model for the evolution of black holes, which is more suited for the methods we employ. In this model, which we call \textit{local random quantum circuit model}, the qubits are arranged in a circle and in each step a random two-qubit gate is applied to two neighbouring qubits. 

\begin{theorem} \label{3design}
$5 n \log(1/\varepsilon)$-size local random quantum circuits form an $\varepsilon$-approximate unitary 3-design. 
\end{theorem}

The proof of Theorem \ref{3design} is based on a reduction, first put forward by Brown and Viola \cite{BV10}, connecting the convergence rate of moments of the random quantum circuit to the spectral gap (the difference of the lowest and second lowest eigenvalues) of a quantum local Hamiltonian. Our main contribution is to show in section \ref{proof3design}  that we can obtain a lower bound on this spectral gap employing a technique from quantum many-body theory used  e.g. in \cite{AFH09, FNW92, Nac96, PVWC07}.  

In particular, using this technique we are able to reduce the problem of bounding the spectral gap of the random walk on $n$ qubits induced by the random circuit, to bounding the spectral gap of the same random walk, but now defined only on \textit{three} neighbouring qubits. Then it suffices to bound the convergence time of the second and third moments of the latter random walk in order to prove that the random circuit constitute a 3-design. We believe our approach is promising also for higher values of $t$ and might pave the way to a proof that random quantum circuits are approximate unitary $t$-designs for all $t = \text{poly}(n)$. We however leave such possibility as an open problem for future work.

Combining Theorems \ref{3design}  and \ref{classeswithlargeC} we obtain our main result that almost any polynomial quantum circuit is useful for exponential quantum speed-ups.

\begin{theorem}
For the distribution induced by the local random quantum circuit model, all but a $2^{-\Omega(n)}$ fraction of quantum circuits $U$ with more than $O(n^{2})$ gates are such that $\U$-$\CC$ shows an exponential gap in the quantum and the postselected classical query complexities.
\end{theorem}

\vspace{0.1 cm}
\noindent \textbf{The role of $C(U)$ and classical efficient solution for sparse unitaries:} We have seen that dispersive unitaries $U$ with large $C(U)$ give an exponential speed-up in U-$\CC$. Is a large $C(U)$ always required for a speed-up? We present two results indicating that this is indeed the case.


First we show that with a modified notion of oracle access (we call it the \textit{independent query model}), in which a different independent realization of the random parameters of the oracle is chosen in each query, a linear $C(U)$ is necessary for an exponential speed-up. 

\begin{theorem} \label{independentquerymodeltheorem}
In the independent query model of oracle access, the randomized classical query complexity of $\U$-$\CC$ is equal to $2^{\Theta(C(U))}$.
\end{theorem}

Second we consider the circuit checking problem for \textit{approximately-sparse} $U$, defined as unitaries which can be approximated (in operator norm) by a sparse matrix with only polynomially many non-zero entries in each row and column. Then we show the following. 

\begin{theorem} \label{classsumulasparseU}
For approximately-sparse $U$ the randomized classical query complexity of $\U$-$\CC$ is polynomial. 
\end{theorem}

We prove Theorem \ref{classsumulasparseU} by showing how a recent result of Van den Nest \cite{VdN09} on the classical simulability of certain quantum states and operations implies that the quantum algorithm for  U-$\CC$ with a sparse $U$ can be efficiently simulated with only polynomial many classical queries to $f$ and $g$. 

\subsection{Related Work}

This paper has a similar flavor to Hallgren and Harrow's work on the $\RFS$ problem \cite{HH08}. The idea of considering the dispersiveness of quantum circuits as a resource for oracle speed-ups also first appeared in \cite{HH08}, where a different, but related, notion of dispersive circuits was proposed and a constant versus linear separation in query complexity was shown for all such dispersive circuits; in section \ref{flatcir} we discuss it in more detail and show that our definition of a dispersive circuit is somewhat more demanding than theirs (although not completely comparable). In \cite{HH08} it was shown that both the Fourier transform over any finite group and almost any sufficiently long quantum circuit are dispersive. Implicit in their work is also the statement that most elements of an approximate unitary 2-design are dispersive. Although their definition of dispersiveness is weaker than ours and therefore broader, the separations we obtain are much stronger. While we get an exponential separation of quantum and postselected classical query complexities, they get a superpolynomial versus polynomial separation of quantum and classical query complexities, and only by using recursion (which itself can be seen as the responsible for the superpolynomial speed-up). 
 
There has been a series of work \cite{DCEL09, GAE07, RS09, HL09, Low09} on unitary $t$-designs (and on the closely related quantum expanders \cite{BST10, Has07, BST07, GE08, Har08, HH09}) and on their connection to random quantum circuits \cite{DCEL09, GAE07, ELL05, ODP07, DOP07, HL09, DJ10, AB08, BV10}. An important problem in this area is to derive efficient constructions on a quantum computer of approximate unitary $t$-designs. While there are several efficient constructions for $2$-designs \cite{DCEL09, BST10, GE08, Har08}, there is only a single one (based on the QFT) for unitary $t$-design on $n$ qubits with $t > 2$ (going up to $t = \Omega(n / \log(n))$) \cite{HL09b}. Our proof that random quantum circuits constitute a $3$-design gives an alternative efficient construction for the $t = 3$ case. 

Recently Brown and Viola \cite{BV10} proposed an interesting approach to the problem of random quantum circuits as unitary $t$-designs, based on mapping the convergence time of moments of the random circuit to the spectral gap of a mean-field quantum Hamiltonian. Conditioned on an unproven, but reasonable, conjecture about the low-lying eigenstates of the Hamiltonian, they showed that random quantum circuits of linear length are $t$-designs for every fixed $t$ and sufficiently large $n$. Our approach also starts by a reduction of the problem to lower bounding the spectral gap of quantum Hamiltonians. However, in our case, we find a local quantum Hamiltonian, consisting of nearest-neighbor terms only. We are also able to rigorously lower bound such spectral gap for $t \leq 3$, therefore obtaining a complete proof in this case.

After the completion of our work, we learned about a recent paper by Fefferman and Umans \cite{FU10}, in which the problem U-$\CC$ is also considered. Their focus is to study the usefulness of this problem in constructing an oracle separation of $\BQP$ and $\PH$, by relating such possibility to a conjecture \cite{BSW03} about the capacity of the Nisan-Wigderson pseudorandom generator \cite{NW94} to fool $AC_0$. To this aim only unitaries of a very special structure are consired in U-$\CC$. Our approach has the advantage that we can show an exponential gap of quantum and postselected-classical query complexities for a generic polynomial quantum circuit (a task not considered in \cite{FU10}), but has the drawback that we fail to give evidence that there are circuits providing a separation of $\BQP$ to $\PH$.

\section{Preliminaries}

\subsection{The oracle problem and its quantum solution} \label{oracleproblem}

Given a unitary $U \in \mathbb{U}(N)$ (with $\mathbb{U}(N)$ the group of $N \times N$ unitary matrices) we consider the following extension of the $\FC$ problem \cite{Aar09}:

\begin{itemize}
\item[] $\U$-$\CC$: We are given access to two black-box functions\footnote{In this work we consider the phase-oracle model for quantum queries. Namely, let $U_f$ be the oracle unitary of $f$ and $\ket{x}$ a computational basis state. Then $U_f \ket{x} = (-1)^{f(x)} \ket{x}$ (and likewise for $g$).} $f, g : \{ 0, 1 \}^{n} \rightarrow \{ 1, -1 \}$ with the promise that either
\item (independent and random) $f$ and $g$ are chosen independently and uniformly at random, with each of their entries drawn from a random unbiased coin, or 
\item ($U$-correlated) for a vector $v \in \mathbb{C}^{N}$ with entries $v_x$ drawn independently from a complex normal distribution $v_x = v_{x, r} + i v_{x, i}$, with $v_{x, r}$ and $v_{x, i}$ normal real variables of mean 0 and variance 1, the functions are chosen as $f(x) = \text{sgn}(\text{Re}(v_x))$ and $g(x) = \text{sgn}(\text{Re}((Uv)_x))$ \footnote{here $\text{Re}(z)$ is the real part of the complex number $z$.}. The vector $Uv$ is given explicitely by  
\begin{equation} 
(Uv)_x = \sum_{y \in [N]} U_{xy}  v_y. 
\end{equation} 
The problem is to decide which is the case.
\end{itemize}

Consider the following quantum algorithm for solving $\U$-$\CC$, where $U$ acts on $n$ qubits:

\begin{tabular}{|l|}
\hline
\textbf{Quantum Algorithm for $\U$-$\CC$:} \\
 (i) Prepare each of the $n$ qubits in the $\ket{+}$ state. \\
 (ii) Query the $f$ oracle. \\
(iii) Apply the unitary $U$. \\
(iv) Query the $g$ oracle.  \\
 (v) Measure each qubit in the Hadamard basis $\{ \ket{+}, \ket{-} \}$ and accept if all qubits are in the $\ket{+}$ state.  \\
\hline
\end{tabular}

\vspace{0.3 cm}

Let
\begin{equation} \label{f}
\ket{f} := \frac{1}{2^{n/2}}\sum_{x \in \{0, 1\}^{n}} (-1)^{f(x)} \ket{x}
\end{equation}
and
\begin{equation} \label{g}
\ket{g} := \frac{1}{2^{n/2}} \sum_{x \in \{0, 1\}^{n}} (-1)^{g(x)} \ket{x}
\end{equation}
Then it follows that the acceptance probability of the algorithm is given by
\begin{equation}
p_U(f, g) :=  | \bra{g}U\ket{f} |^{2}.
\end{equation}
The next proposition shows the quantum algorithm above can distinguish the cases of correlated (by the action of $\U$) and independent $f$ and $g$.
\begin{proposition} \label{quantumalgorithm}
If $f$ and $g$ are drawn independently and uniformily at random,
\begin{equation}
\mathbb{E} \left(p_U(f, g) \right) = \frac{1}{2^{n}} 
\end{equation}
while if $f$ and $g$ are $\U$-correlated,
\begin{equation}
\mathbb{E} \left(p_U(f, g) \right) \geq 0.07
\end{equation}
\end{proposition}
\begin{proof}
In \cite{Aar09} Aaronson proved the proposition for the case in which $U$ is the quantum Fourier transform over $\mathbb{Z}_2^{n}$, which appeared as Theorem 9 in \cite{Aar09}. A closer inspection at his proof shows that the only property of the quantum Fourier used is the fact that it is a unitary. Therefore the reasoning of \cite{Aar09} can be applied here without any modification. We omit reproducing the full argument and instead refer the reader to \cite{Aar09}.
\end{proof}

\subsection{Dispersing Circuits} \label{flatcir}

We now define a notion of dispersive, or flat, circuits which will play a central role in this work. Let $h_{\text{min}}$ be the min-entropy defined as
\begin{equation}
h_{\text{min}}(p) = - \log  \max_{x} p(x). 
\end{equation}
 
\begin{definition} \label{dispersivenesshere}
For a unitary $U$ we define:
\begin{eqnarray}
C(U) &:=&  \min_{j \in [N]} h_{\text{min}} \left( \{ |\bra{0}U\ket{j}|^{2}, ..., |\bra{N-1}U\ket{j}|^{2} \}  \right) \\ \nonumber \\ &=& - \log \left( \max_{i, j \in [N]} |U_{i, j}|^{2} \right).   
\end{eqnarray} 
\end{definition}

It is interesting to compare this definition of a dispersive circuit with Harrow and Hallgren's \cite{HH08}:
\begin{definition} \label{HHdis}
(HH-dispersiveness \cite{HH08}) A unitary $U \in \mathbb{U}(2^{n})$ is $(\alpha, \beta)$-dispersing if there exists a set $A \subseteq \{ 0, 1\}^{n}$ with $|A|\geq 2^{\alpha n}$ and
\begin{equation}
\sum_{x \in \{0, 1\}^{n}} |\bra{a}U\ket{x}| \geq \beta 2^{\frac{n}{2}}
\end{equation}
for all $a \in A$. 
\end{definition}

Thus while Def. \ref{dispersivenesshere} looks at the infinity norm of the outcome probability distribution of measurements in the computational basis, maximized over all initial computational basis states, HH-dispersiveness (Def. \ref{HHdis}) is concerned with the 1-norm of such probability distribution, and the maximum taken only over a constant-size fraction of all the computational basis states. 

In \cite{HH08} it was shown that $(\alpha, \beta)-$dispersing circuits, for $\alpha, \beta = O(1)$, are useful for speed-ups in the variant of the $\RFS$ problem there defined. If we allow for lower values of $\beta$ than a constant (but still requiring the circuit to be fairly flat), then a dispersive unitary according to Def. \ref{dispersivenesshere} is also HH-dispersive. Indeed, a simple calculation (which we omit here) shows that if $U$ is such that $C(U) \geq \gamma n$, then $U$ is also $(1, 2^{(\gamma - 1) n / 2})$-dispersive according to Def. \ref {HHdis}.

\subsection{Approximate Unitary Designs} \label{approximateunitarydesigns}

We start defining a norm on quantum operations which we will use to compare two superoperators. For $X \in {\cal B}(\mathbb{C}^{d})$ we define the $p$-Schatten norms $\Vert X \Vert_p := \tr(|X|^{p})^{\frac{1}{p}}$. Then for a superoperator $\Lambda : {\cal B}(\mathbb{C}^{d}) \rightarrow {\cal B}(\mathbb{C}^{d'})$ we define the $p \rightarrow q$ induced Schatten norm as
\begin{equation}
\Vert {\Lambda}(X) \Vert_{p \rightarrow q}  := \sup_{X \neq 0} \frac{\Vert {\Lambda}(X) \Vert_p}{\Vert X \Vert_q}.
\end{equation}
Finally, the diamond norm is defined as the CB-completion of the $1 \rightarrow 1$ norm,
\begin{equation}
\Vert \Lambda \Vert_{\diamond} := \sup_d \Vert \id_{d} \otimes \Lambda \Vert_{1 \rightarrow 1}.
\end{equation}

There are several different definitions of $\varepsilon$-approximate unitary $t$-designs. A convenient one for us is the following.
\begin{definition} \label{tdesgndef}
\textbf{(Approximate unitary $t$-design)} Let $\{ \mu, U \}$ be an ensemble of unitary operators from $\mathbb{U}(d)$. Define
\begin{equation}
{\cal G}_{\mu, t}(\rho) = \int_{\mathbb{U}(d)} U^{\otimes t}\rho (U^{\cal y})^{\otimes t} \mu(dU)
\end{equation}
and
\begin{equation}
{\cal G}_{H, t}(\rho) = \int_{\mathbb{U}(d)} U^{\otimes t}\rho (U^{\cal y})^{\otimes t} \mu_H(dU),
\end{equation}
where $\mu_{H}$ is the Haar measure. Then the ensemble is a $\varepsilon$-approximate unitary $t$-design if 
\begin{equation} \label{designcondepsilon}
\Vert {\cal G}_{\mu, t} - {\cal G}_{H, t} \Vert_{2 \rightarrow 2} \leq \varepsilon.
\end{equation}
\end{definition}

The following Lemma from \cite{Low10} shows that the previous notion of an approximate unitary design implies two others, which will also be relevant in this work.
\begin{lemma} \label{equilow}
(Lemma 2.2.14 of \cite{Low10}). Let $\{ \mu, U \}$ be an $\varepsilon$-approximate unitary $t$-design on $\mathbb{U}(d)$ according to Def. \ref{tdesgndef}. Then 
\begin{itemize}
\item[] (a) For $ {\cal G}_{\mu, t}$ and $ {\cal G}_{H, t}$ given by Def. \ref{tdesgndef},
\begin{equation}
\Vert {\cal G}_{\mu, t} - {\cal G}_{H, t} \Vert_{\diamond} \leq d^{t} \varepsilon.
\end{equation}
\item[] (b) For every balanced monomial $M = U_{p_1q_1}...U_{p_kq_k}U_{r_1s_1}^{*}...U_{r_ks_k}^{*}$ of degree $k \leq t$,
\begin{equation}
| \mathbb{E}_{U \sim \mu} \left(  M(U) \right) - \mathbb{E}_{U \sim \mu_H} \left(  M(U) \right)| \leq d^{2t} \varepsilon.
\end{equation}
\end{itemize}
\end{lemma}

\section{Families of Flat Unitaries} \label{familiesflat}

In this section we prove Theorem \ref{classeswithlargeC} showing two examples of families of unitaries which are highly dispersing.

\noindent\textbf{Quantum Fourier Transforms:} Let $G$ be a finite group with irreducible unitary representations $\{  V_{\lambda} \}_{\lambda \in \widehat{G}}$, where $r_{\lambda}(g)$ is the unitary matrix representation of the group element $g \in G$ in the irrep $V_{\lambda}$ and $\widehat{G}$ labels all inequivalent irreps of $G$. Let also $\{ \ket{g} \}_{g \in G}$ be an orthogonal basis for $\mathbb{C}^{|G|}$. The quantum Fourier transform over $G$ is given by 
\begin{equation} \label{QFToverG}
U_{\text{QFT}}(G) = \sqrt{\frac{\dim V_{\lambda}}{|G|}} \sum_{g \in G} \sum_{\lambda \in \widehat{G}} \sum_{i, j=1}^{\dim V_{\lambda}} r_{\lambda}(g)_{ij} \ket{\lambda, i, j}\bra{g}.
\end{equation}

We now prove the first part of Theorem \ref{classeswithlargeC} .

\begin{proof}
(Theorem \ref{classeswithlargeC} part (i)) The statement is a simple application of the following basic relation, valid for any finite group \cite{Bur93}:
\begin{equation} \label{eqbasicgrepteheoru}
\sum_{\lambda \in \widehat{G}} \dim(V_{\lambda})^{2} = |G|.
\end{equation}
Indeed, Eq. (\ref{eqbasicgrepteheoru}) implies $\dim V_{\lambda} \leq |G|^{\frac{1}{2}}$ for every $\lambda \in \widehat{G}$ and thus
\begin{equation}
|\bra{\lambda, i, j} U_{\text{QFT}}(G)  \ket{g}|^{2} = \frac{\dim V_{\lambda}}{|G|} |r_{\lambda}(g)_{ij}|^{2} \leq |G|^{-\frac{1}{2}},
\end{equation}
which implies $C(U_{\text{QFT}}(G)) \geq \frac{\log |G|}{2}$. 
\end{proof}
\vspace{0.4 cm}
\noindent\textbf{Unitary 3-designs:}

\begin{lemma} \label{lemmatseidng}
For every $\varepsilon$-approximate $t$-design $\{ \mu(dU), U \}$ on $\mathbb{U}(d)$, 
\begin{equation}
\Pr_{U \sim \mu} \left( C(U) \leq \nu  \right) \leq d^{2}2^{\nu t}\left(d^{-t}t! + d^{2t}\varepsilon \right).
\end{equation}
\end{lemma}
\begin{proof}
By Markov's inequality
\begin{eqnarray} \label{markov}
\Pr_{U \sim \mu} \left( |\bra{i}U\ket{j}|^{2} \geq \lambda  \right) &=& \Pr_{U \sim \mu} \left( |\bra{i}U\ket{j}|^{2t} \geq \lambda^{t}  \right) \nonumber \\ &\leq & \frac{ \mathbb{E}_{U \sim \mu} \left(  |\bra{i}U\ket{j}|^{2t} \right) }{\lambda^{t}}.
\end{eqnarray}
Since $\{ \mu(dU), U \}$ is an $\varepsilon$-approximate unitary $t$-design, Lemma \ref{equilow} gives
\begin{equation} \label{bounddesign}
\mathbb{E}_{U \sim \mu} \left(  |\bra{i}U\ket{j}|^{2t} \right) \leq \mathbb{E}_{U \sim \mu_H} \left(  |\bra{i}U\ket{j}|^{2t} \right)  + d^{2t} \varepsilon.
\end{equation}

We have
\begin{eqnarray} \label{averageHaar}
\mathbb{E}_{U \sim \mu_H} \left(  |\bra{i}U\ket{j}|^{2t} \right) &=& \mathbb{E}_{U \sim \mu_H} \left(  (\bra{i}U\ket{j}\bra{j}U^{\cal y}\ket{i})^{t} \right) \nonumber \\ &=& \tr \left( \int_{U} \mu_{H}(dU) U^{\otimes t} \ket{j}\bra{j}^{\otimes t} (U^{\cal y})^{\otimes t} \left( \ket{i}\bra{i} \right)^{\otimes t}  \right) \nonumber \\ &=& \binom{d+t-1}{t}^{-1} \tr \left(  P_{\text{sym}, t}  \left( \ket{i}\bra{i} \right)^{\otimes t}  \right) =  \binom{d+t-1}{t}^{-1},
\end{eqnarray}
for $P_{\text{sym}, t}$ the projector onto the $ \binom{d+t-1}{t}$-dimensional symmetric subspace of $({\cal C}^{d})^{\otimes t}$.

Then, from Eqs. (\ref{markov}), (\ref{bounddesign}), (\ref{averageHaar}) and the union bound,
\begin{equation}
\Pr_{U \sim \mu} \left( \max_{i, j \in [d]} |\bra{i}U\ket{j}|^{2} \geq \lambda  \right) \leq  \frac{d^{2}}{\lambda^{t}}  \left( \binom{d+t-1}{t}^{-1} + d^{2t} \varepsilon \right).
\end{equation}
Now we set $\lambda = 2^{- \nu}$, use the bound 
\begin{equation}
\binom{d+t-1}{t} = \frac{(d + t -1)...(d+1)d}{t!} \geq \frac{d^{t}}{t!},
\end{equation}
and we are done.
\end{proof}

\begin{proof} (Theorem \ref{classeswithlargeC} part (ii))

Let $\{ \mu(dU), U \}$ be a $2^{-3tn}$-approximate unitary $t$-design on $\mathbb{U}(2^{n})$. Then applying Lemma \ref{lemmatseidng} with $\nu = \beta n$, $\Pr_{U \sim \mu} \left( C(U) \leq \beta n  \right) \leq 2^{2n} 2^{\beta n t} 2^{-t n + 1}(1 + t!)$.
\end{proof}

\section{Classical lower bounds} \label{classicallowerbounds}

In this section we prove an exponential lower bound on the postselected classical query complexity of $\U$-$\CC$ for dispersive circuits. Following \cite{Aar09}, our strategy will be to show that the distribution in the $\U$-$\CC$ problem in the case of $\U$-correlated strings is approximately $k$-wise independent. Then the result follows from the following proposition from \cite{Aar09}, relating this property to bounds on the postselected query complexity of distinguishing such distribution from the uniform one:

\begin{proposition} \label{connectingalmostwiseindtolowerbounds}
(Lemma 20 of \cite{Aar09}) Suppose a probability distribution ${\cal D}$ over oracle strings is $\delta$-almost $k$-wise independent. Then no bounded-error postselected classical machine running in less than $k$ steps can distinguish ${\cal D}$ from the uniform distribution with bias larger than $2 \delta$. 
\end{proposition}

\begin{proof} (Theorem \ref{maindesignals})
Propositions \ref{almostwiseindepdnent} and \ref{connectingalmostwiseindtolowerbounds} give a lower bound of $2^{C(U)/7}$ on the classical query complexity with postselection of U-$\CC$. Together with the $O(1)$ queries quantum algorithm for the problem from section \ref{oracleproblem} implies Theorem \ref{maindesignals}.
\end{proof}

For a string $\{ x_1, ..., x_M \} \in \{  -1, 1\}^{M}$ we call a term of the form $\frac{1 \pm x_i}{2}$ a literal and define a $k$-term as a product of $k$-literals, which equals 1 if all the literals are 1 and 0 otherwise. Then an approximate $k$-wise independent distribution is defined as follows. 

\begin{definition}
A distribution ${\cal D }$ over $\{ -1, 1\}^{M}$ is $\varepsilon$-almost $k$-wise independent if for every $k$-term $C$,
\begin{equation}
\frac{1 - \varepsilon}{2^{k}} \leq \Pr_{\cal D}(C) \leq \frac{1 + \epsilon}{2^{k}}.
\end{equation}
\end{definition}
In words, ${\cal D}$ is $\varepsilon$-almost $k$-wise independent if the probability of every $k$-term is $\varepsilon$-multiplicatively close to its value on the uniform distribution (which is simply $2^{-k}$).

Consider the vector $\omega_{v, U} \in \{ -1, 1\}^{2N}$ given by
\begin{equation}
\omega_{v, U} := ( \text{sgn}(v_1), ..., \text{sgn}(v_N), \text{sgn}(\text{Re}(Uv)_1), ..., \text{sgn}(\text{Re}(Uv)_N))
\end{equation}
and let ${\cal D}_{U}$ be the distribution over $\omega_{v, U}$ when the vector $v := (v_1, ..., v_N)$ is composed of independent entries $v_k$, each drawn from a complex normal distribution $v_{k} = v_{k, r} + i v_{k, i}$ with $v_{k, r}, v_{k, i}$ real normal variables of mean 0 and variance 1. Then we have

\begin{proposition} \label{almostwiseindepdnent}
${\cal D}_U$ is $(6k^{3}2^{-C(U)/2})$-almost $k$-wise independent. 
\end{proposition}
\begin{proof}
Define
\begin{equation}
z_j = 
\begin{cases}
v_{s(j)} & \text{if} \hspace{0.2 cm} 0 \leq j \leq m \\
(Uv)_{r(j)} & \text{if} \hspace{0.2 cm} m < j \leq k.
\end{cases}
\end{equation}
for injective functions $s : [m] \rightarrow [N]$ and $r : [k - m] \rightarrow [N]$. Consider the following probability 
\begin{equation} \label{probinde}
P := \Pr \left( \text{sgn}(\text{Re}(z_1)) = a_1, \text{sgn}(\text{Re}(z_2)) = a_2, ..., \text{sgn}(\text{Re}(z_k)) = a_k \right),
\end{equation}
for a tuple $a := (a_1, ..., a_k) \in \{ -1, 1 \}^{k}$. In the remainder of the proof we show that the probability in Eq. (\ref{probinde}) is $(6k^{3}2^{-C(U)/2})$-multiplicatively close to $2^{-k}$ for every choice of the tuple $a$, functions $s, r$, and integer $m \leq k$, which readily implies the statement of the proposition.

First we note that the probability of Eq. (\ref{probinde}) is equal to
\begin{equation} \
P = \Pr \left(  a_1 \text{Re}(z_1) \geq 0, a_2 \text{Re}(z_1) \geq 0, ..., a_k \text{Re}(z_k) \geq 0  \right)
\end{equation}
and that $y := (a_1 z_1, ..., a_k z_k)$ is a (circular symmetric complex) multivariable normal distribution with mean zero. That is, there is a matrix $M \in \mathbb{C}^{k \times N}$ such that $y  := Mv$, where $v \in \mathbb{R}^{N}$ is a vector of independent normal variables of mean 0 and variance 1. 

It is a standard fact of multivariate normal distributions that they are are completely specified by the mean vector $\mu$ and the covariance matrix $\Sigma$ of the distribution, i.e. the probability density function of $y$ is given  by
\begin{equation}
\frac{1}{\pi^{k}|\Sigma|^{k}}  \text{exp} \left( - y^{\cal y}\Sigma^{-1}y  + \mu^{\cal y} y \right).
\end{equation}
In our case $\mu = 0$, while the covariance matrix $\Sigma_{i,j} = \mathbb{E} \left( y_{i}^{*}y_j \right)$ is given by
\begin{equation}
\Sigma = \left(   
\begin{array}{cc}
\id_{m} & Q \\
Q^{\cal y} & \id_{k-m}
\end{array}
\right),
\end{equation}
with
\begin{equation}
Q_{i, j} = a_i a_j\mathbb{E} \left( v_{s(i)}  (Uv)_{r(j)}    \right) = a_i a_j U_{s(i)r(j)}.
\end{equation}
Thus $\Vert \Sigma - \id \Vert_{\infty}  \leq k^{2} \max_{i, j \in [N]} |U_{i, j}| \leq k^{2} 2^{-\frac{1}{2}C(U)}$, which implies 
\begin{equation}
(1 - k^{2} 2^{-\frac{1}{2}C(U)}) \id   \leq  \Sigma \leq (1 + k^{2} 2^{-\frac{1}{2}C(U)}) \id
\end{equation}
and
\begin{equation} \label{sigmaiverse}
 (1 + k^{2} 2^{-\frac{1}{2}C(U)})^{-1} \id   \leq  \Sigma^{-1} \leq (1 - k^{2} 2^{-\frac{1}{2}C(U)})^{-1} \id.
\end{equation}
For $k$ much smaller than $2^{C(U)}$ we thus see that the covariance matrix $\Sigma$ is close to the identity, which means that the distribution over $y$ is close to the uniform. In the rest of the proof we make this observation quantitative.

We have
\begin{eqnarray} \label{integrals}
P &=& \Pr \left(  \text{Re}(y_1) \geq 0, \text{Re}(y_2) \geq 0, ..., \text{Re}(y_k) \geq 0  \right) \nonumber \\ &=& |\Sigma|^{- k} \pi^{- k} \int_{\text{Re}(y_1) \geq 0, ..., \text{Re}(y_k) \geq 0} \text{exp} \left( -  y^{\cal y}\Sigma^{-1}y  \right) dy_1 ... dy_k  \nonumber \\ &\leq& |\Sigma|^{- k} \pi^{- k} \int_{\text{Re}(y_1) \geq 0, ..., \text{Re}(y_k) \geq 0} \text{exp} \left( - (1 + k^{2} 2^{-\frac{1}{2}C(U)})^{-1} y^{\cal y}y  \right) dy_1 ... dy_k\nonumber \\  &\leq &  \left( \frac{ 1 + k^{2} 2^{-\frac{1}{2}C(U)}}{ 1 - k^{2} 2^{-\frac{1}{2}C(U)})} \right)^{k} \pi^{-k} \int_{\text{Re}(y_1) \geq 0, ..., \text{Re}(y_k) \geq 0} \text{exp} \left( - y^{\cal y}y  \right) dy_1 ... dy_k \nonumber \\  &= &  \left( \frac{ 1 + k^{2} 2^{-\frac{1}{2}C(U)}}{ 1 - k^{2} 2^{-\frac{1}{2}C(U)})} \right)^{k} 2^{-k},
\end{eqnarray}
where the two inequalities in the Eq. (\ref{integrals}) above follow from the two sides of Eq. (\ref{sigmaiverse}). Then using the bound 
\begin{equation} \label{boundeasy}
\left( \frac{1+a}{1-a} \right)^{k} \leq 1 + 6ka, 
\end{equation}
valid for all $0 \leq a \leq 1$, we get
\begin{equation}
P \leq (1 + 6k^{3}2^{-C(U)/2}) 2^{-k}.
\end{equation}

By a completely similar argument (using again the two sides of Eq. (\ref{sigmaiverse}) and Eq. (\ref{boundeasy})) we also find 
\begin{equation}
P \geq  \left( \frac{ 1 - k^{2} 2^{-\frac{1}{2}C(U)}}{ 1 + k^{2} 2^{-\frac{1}{2}C(U)})} \right)^{k} 2^{-k} \geq (1 - 6k^{3}2^{-C(U)/2}) 2^{-k}, 
\end{equation}
and we are done.
\end{proof}

\section{Classical upper bounds} \label{upperbounds}

In this section we prove Theorems \ref{independentquerymodeltheorem} and \ref{classsumulasparseU}.

\begin{proof} \textbf{(Theorem \ref{independentquerymodeltheorem})}

Let $(i, j)$ be such that $|U_{i, j}| = 2^{- C(U)/2}$. Assuming that $\text{Re}(U_{i, j}) \geq \text{Im}(U_{i, j})$ \footnote{Otherwise we consider instead the unitary $\sqrt{-1}U$}, we get $\text{Re}(U_{i, j}) \geq 2^{- C(U)/2 + 1}$.

The algorithm for U-$\CC$ in the independent query model works as follows: One queries $f_r(i)$ and $g_r(j)$ over $N = \lceil \text{Re}(U_{i, j})^{-2} \log(1/\epsilon) \rceil$ independent realizations of the oracle (labelled by $r$) and computes 
\begin{equation}
E_N := \frac{1}{N} \sum_{r=1}^{N} f_r(i)g_r(j),
\end{equation}
deciding that the functions are $\U$-correlated if $E_N \geq \text{Re}(U_{i, j})/4$.

In the case of independent $f$ and $g$, we have $\mathbb{E}(f(i)g(j)) = 0$. In the remainder of the proof we show that for $\U$-correlated $f$ and $g$, $\mathbb{E}(f(i)g(j)) \geq \text{Re}(U_{i, j})/2$. Then Chernoff bound gives that the algorithm fails with probability at most $\epsilon$.

Let us turn to the lower bound on $\mathbb{E}(f(i)g(j))$ in the case of $\U$-correlated $f$ and $g$. We have
\begin{equation} \label{productfgdicrete}
\mathbb{E}(f(i)g(j)) = \mathbb{E}(\text{sgn}(\text{Re}(v_i)) \text{sgn}(\text{Re}(Uv)_j)).
\end{equation}
Note that $\mathbb{E}( \text{Re}(v_i) \text{Re}(Uv)_j)) = \text{Re}(U_{i, j}) \geq 2^{-(C(U)/2 + 1)}$. So all we have to do is to check that the discretized version given by Eq. (\ref{productfgdicrete}) has a similar expectation value.

First we write
\begin{eqnarray}
\mathbb{E}(\text{sgn}(\text{Re}(v_i)) \text{sgn}(\text{Re}(Uv)_j)) &=& \Pr \left( \text{Re}(v_i) \text{Re}(Uv)_j \geq 0 \right) -  \Pr \left( \text{Re}(v_i) \text{Re}(Uv)_j < 0 \right) \nonumber \\ &=& 2 \Pr \left( \text{Re}(v_i) \text{Re}(Uv)_j \geq 0 \right) - 1 \nonumber \\ &=& 4 \Pr \left( \text{Re}(v_i) \geq 0 \hspace{0.2 cm} \text{and} \hspace{0.2 cm} \text{Re}(Uv)_j \geq 0  \right) - 1.
\end{eqnarray}
Now define the bivariate normal variable $w := (w_1, w_2)$ with $w_1 = v_i$ and $w_2 := (Uv)_j$. The probability distribution of $w$ is completely characterized by the covariance matrix with entries $\Sigma_{k,l} := \mathbb{E}(w_k w_l^{*})$, which reads 
\begin{equation}
\Sigma = \left(   
\begin{array}{cc}
1 & U_{i, j} \\
U_{i, j}^{*} & 1
\end{array}
\right).
\end{equation}
Then,
\begin{eqnarray}
&&\Pr \left( \text{Re}(v_i) \geq 0 \hspace{0.2 cm} \text{and} \hspace{0.2 cm} \text{Re}(Uv)_j \geq 0  \right) \nonumber \\ &=& |\Sigma|^{-2}\pi^{-2} \int_{\text{Re}(w_1) \geq 0, \text{Re}(w_2) \geq 0} \text{exp}\left( - w^{\cal y}\Sigma^{-1}w \right)dw_1 dw_2 \nonumber \\ &=& |\Sigma|^{-2}\pi^{-2} \int_{\text{Re}(w_1) \geq 0, \text{Re}(w_2) \geq 0} \text{exp}\left( - (1 - |U_{i, j}|)^{-1}w^{\cal y}w + 2\text{Re}(w_1^{*}w_2U_{i, j}) \right)dw_1 dw_2 \nonumber \\  &\geq& |\Sigma|^{-2}\pi^{-2} \int_{\text{Re}(w_1) \geq 0, \text{Re}(w_2) \geq 0} \text{exp}\left( - (1 - |U_{i, j}|)^{-1}w^{\cal y}w \right) \left(1 + 2\text{Re}(w_1^{*}w_2U_{i, j}) \right)dw_1 dw_2 \nonumber \\ &=& \pi^{-2} \int_{\text{Re}(w_1) \geq 0, \text{Re}(w_2) \geq 0} \text{exp}\left( - w^{\cal y}w \right) \left(1 + 2\text{Re}(w_1^{*}w_2U_{i, j})(1 - |U_{i, j}|)^{2} \right)dw_1 dw_2 \nonumber \\ &=& \frac{1}{4} \left(1 + \frac{\text{Re}(U_{i, j})(1 - |U_{i, j}|)^{2}}{8} \right).
\end{eqnarray}
where we used basic facts of Gaussian integrals and that $e^{x} \geq 1 + x$, for $x \geq 0$.
\end{proof}

\vspace{0.5 cm}

We now turn to the proof of Theorem \ref{classsumulasparseU}, which is largely based on recent techniques of Van den Nest for the efficient classical simulation of certain types of quantum states and operations \cite{VdN09}. We will make use of the notion of \textit{computational tractable} states, which are defined below in a slightly more general way than in \cite{VdN09}, in order to accommodate for oracle queries.
\begin{definition}
\cite{VdN09} A state on $n$ qubits $\ket{\psi}$ is $f$-computational tractable given access to the oracle function $f_{\psi} : \{ 0, 1 \}^{m} \rightarrow \{ 0, 1 \}$ (with $m = \poly(n)$) if the following holds
\begin{list}{\quad}{}
\item[] (a) it is possible to sample from the probability distribution $\Pr(x) := |\braket{x}{\psi}|^{2}$ on the set of $n$-bit strings  in $\poly(n)$ time in a classical computer with $\poly(n)$ many queries to $f_{\psi}$, and
\item[] (b) upon input of any bit of $x$, the coefficient $\braket{x}{\psi}$ can be computed in $\poly(n)$ time on a classical computer with $\poly(n)$ queries to $f_{\psi}$.
\end{list}
\end{definition}

Then we have:
\begin{proposition} \label{vandennest}
\cite{VdN09} Let $\ket{\psi}$ and $\ket{\phi}$ be $f$-computational tractable states (given access to oracle $f$) of $n$ qubits each and let $A$ be an efficiently computable sparse $n$-qubit operation with $\Vert A \Vert_{\infty} \leq 1$. Then there exists an efficient classical algorithm to approximate $\bra{\phi}A \ket{\psi}$ with polynomial accuracy in $n$, given access to the oracle $f$.
\end{proposition}
In \cite{VdN09} Van den Nest proved Proposition \ref{vandennest} in the non-oracular case. However, it is easy to check that his proof carries through without any modification to cover the statement of Proposition  \ref{vandennest}.

\begin{proof} \textbf{(Theorem \ref{classsumulasparseU})}

Let $\tilde U$ be the sparse approximation of $U$ with only $\poly(n)$ many non-zero entries in each row and column and such that $\Vert U - \tilde U \Vert_{\infty} \leq 0.03$. Let $\ket{f}$ and $\ket{g}$ be the states given by Eqs. (\ref{f}) and (\ref{g}). Note that $\ket{f}$ and $\ket{g}$ are $f$- and $g$-computational tractable, respectively. Indeed any of their coefficients can be read directly from the oracles $f$ and $g$, while the probability distribution $\Pr(x) := |\braket{x}{f}|^{2}$ (and analogously $|\braket{x}{g}|^{2}$) is uniform and therefore easily samplable. 

From Proposition \ref{quantumalgorithm} we see it suffices to calculate $|\bra{g}U\ket{f}|^{2}$ to accuracy $< 0.07$ in order to solve U-$\CC$ with high probability and, indeed, 
\begin{eqnarray}
||\bra{g}U\ket{f}|^{2} - |\bra{g}\tilde U\ket{f}|^{2}| &=& |\bra{f}\left( U \ket{g}\bra{g}U^{\cal y} - \tilde U \ket{g}\bra{g} \tilde U^{\cal y}  \right) \ket{f}| \nonumber \\ &\leq& \left \Vert U \ket{g}\bra{g}U^{\cal y} - \tilde U \ket{g}\bra{g} \tilde U^{\cal y} \right \Vert_{1} \nonumber \\ &=&  \left \Vert U \ket{g} - \tilde U \ket{g} \right\Vert . \left \Vert U \ket{g} + \tilde U \ket{g} \right \Vert  \nonumber \\ &\leq&  2 \left \Vert U - \tilde U \right \Vert_{\infty} \leq 0.06 .
\end{eqnarray}
\end{proof}

\section{Random Circuits are approximate $3$-designs} \label{proof3design}

In this section we prove that random quantum circuits of linear length form an approximate unitary 3-design, which is the main technical contribution of the paper. 

We consider two classes of random quantum circuits, both defined as random walks on $\mathbb{U}(2^{n})$:
\begin{itemize}
\item \textit{uniform random circuit}: in each step two indices $i \neq j$ are chosen uniformly at random from $[n]$ and a two-qubit unitary gate $U_{i, j}$ drawn from the Haar measure on $\mathbb{U}(4)$ is applied to qubits $i$ and $j$.
\item \textit{local random circuit}:  in each step of the walk an index $i$ is chosen uniformly at random from $[n]$ and a two-qubit gate $U_{i, i+1}$ drawn from the Haar measure on $\mathbb{U}(4)$ is applied to the two neighbouring qubits $i$ and $i+1$ (we arrange the qubits on a circle, so we identify the $(n+1)$-th qubit with the first). 
\end{itemize}

Throughout this section we will focus on local random circuits and then show how our results can be extended to uniform random circuits.

We will make use the folllowing well-known correspondence of superoperators and operators, which allow us to evaluate the eigenvalues of the former by computing the eigenvalues of the latter. Given a superoperator ${\cal G}$ given by
\begin{equation} \label{Gmatrix}
{\cal G}(X) := \sum_k A_k X B_k^{\cal y},
\end{equation}
we define the operator
\begin{equation}
G := \sum_k A_k \otimes \overline{B}_k,
\end{equation}
with $\overline{X}$ the complex conjugate of $X$. 

Let $X$ be such that $\tr(XX^{\cal y}) = 1$ and ${\cal G}(X) = \lambda X$, for a complex number $\lambda$, i.e. $X$ is an eigenoperator of ${\cal G}$ with eigenvalue $\lambda$. Then defining $\ket{X} := X \otimes \id \ket{\Phi}$, with
\begin{equation}
\ket{\Phi} := \sum_{k} \ket{k} \otimes \ket{k},
\end{equation}
it holds that $G \ket{X} = \lambda \ket{X}$, i.e. $\ket{X}$ is an eingenvector of $G$ with eingenvalue $\lambda$. 

A direct implication of this correspondence is that
\begin{equation} \label{2normsuperinftynormregular}
\Vert {\cal G} \Vert_{2 \rightarrow 2} = \Vert G \Vert_{\infty}. 
\end{equation}

\begin{proof} (Theorem \ref{3design})

Let $\{ \mu(dU), U \}$  be the distribution of unitaries after one step of the random walk according to the local random circuit model. Following Eq. (\ref{Gmatrix}), define
\begin{equation}
G_{\mu^{* k}, t} :=  \int_{\mathbb{U}(d)} \mu^{* k}(dU)  U^{\otimes t} \otimes  \overline{U}^{\otimes t}.
\end{equation}
where $\mu^{* k}$ is the $k$-fold convolution of $\mu$ with itself, i.e. 
\begin{equation}
\mu^{* k} := \int \delta_{U_1...U_k} \mu(dU_1)... \mu(dU_k).
\end{equation}
In the sequel we show that for $t = 2, 3$,
\begin{equation} \label{boundinftynorm}
\Vert G_{\mu^{* k}, t} - G_{\mu_H, t} \Vert_{\infty} \leq \left( 1 - \frac{1}{5n} \right)^{k},
\end{equation}
with $\mu_{H}$ is the Haar measure on $\mathbb{U}(2^{n})$. Then by Def. \ref{tdesgndef} and Eq. (\ref{2normsuperinftynormregular})  we find that $\{ \mu^{* 5n \log(1/\epsilon)}(dU), U \}$ is an $\epsilon$-approximate unitary 3-design, which is the statement of the theorem.

Following \cite{BV10}, we have  $G_{\mu^{* k}, t}  = \left(M_{t, n}\right)^{k}$, with
\begin{equation}
M_{t, n} := \frac{1}{n}\sum_{i=1}^{n} P_{i, i + 1} 
\end{equation}
and
\begin{equation}
P_{i, i+1} := \int_{\mathbb{U}(4)} \mu_H(dU)  U_{i, i+1}^{\otimes t} \otimes  \overline{U}_{i, i+1}^{\otimes t}.
\end{equation}
Moreover the projector onto the eigenvalue-one subspace of $M_t$ is equal to $G_{\mu_H, t}$, since \cite{HL09, BV10}
\begin{equation}
\lim_{k \rightarrow \infty} \Vert G_{\mu^{* k}, t} - G_{\mu_H, t} \Vert_{\infty}  = 0.
\end{equation}
Therefore
\begin{equation} \label{secondgaponG}
\Vert G_{\mu^{* k}, t} - G_{\mu_H, t} \Vert_{\infty} = \left(\lambda_2 \left( M_{t, n} \right)\right)^{k},
\end{equation}
with $\lambda_k(X)$ the $k$-th largest eigenvalue of $X$. 

Lemma \ref{gaponnqubitsfromgapon3qubits} and Lemma \ref{gapequals710} gives that $\lambda_2(M_{t, n}) \leq 1 - \frac{1}{5n}$. Then Eq. (\ref{boundinftynorm}) is a consequence of this bound and Eq. (\ref{secondgaponG}). 
\end{proof}

\vspace{0.3 cm}
The next proposition is the key part of the argument. It shows that in order to upper bound $\lambda_{2}(M_{t, n})$ it is enough to obtain a sufficiently strong upper bound on $\lambda_{2}(M_{t, 3})$. The latter is associated to the convergence time of a random walk on only three qubits and, therefore, can be more easily analysed. 
\begin{lemma} \label{gaponnqubitsfromgapon3qubits}
\begin{equation}
\lambda_{2} \left( M_{t, n} \right) \leq 1 - \frac{3 - 4 \lambda_2 \left(  M_{t, 3} \right)}{n}.
\end{equation}
\end{lemma}
\begin{proof}
Define $H_{i, i+1} := \id - P_{i, i + 1}$ and
\begin{equation} \label{defH}
H := \sum_{i=1}^{n} H_{i, i+1} = n(\id - M_{t, n}).
\end{equation}

The operator $H$ is a quantum local Hamiltonian (i.e. a sum of terms which act non-trivially only on neighbouring sites), composed of local projectors $H_{i, i+1}$, with the following properties: 
\begin{itemize}
\item\textit{ periodic boundary conditions}: the $(i +1)$-th site is identified with the first.
\item \textit{zero gound-state energy}: $\lambda_{\min}(H) = 0$, with $\lambda_{\min}(H)$ the minimum eigenvalue of $H$. 
\item\textit{ frustation-freeness}:  every state $\ket{\psi}$ in the grounstate manifold, composed of all eigenvectors with eigenvalue zero, is such that $H_{i, i+1}\ket{\psi} = 0$, for all $i$. 
\end{itemize}

Let $\Delta(X)$ be the \textit{spectral gap} of $X$, i.e. the difference of the second lowest to the lowest eigenvalues. Then from Lemma \ref{boundongapofromlocal} we get
\begin{eqnarray}
\lambda_{2} \left( M_{t, n} \right) &=& 1 - \frac{\Delta(H)}{n} \nonumber \\ 
&\leq& 1 - \frac{2\Delta \left( H_{1, 2} + H_{2,3} \right) - 1 }{n} \nonumber \\ 
&= & 1 - \frac{4\Delta \left( \id - M_{t, 3} \right) - 1 }{n} \nonumber \\ 
&=& 1 - \frac{3 - 4 \lambda_2 \left( M_{t, 3} \right)}{n}.
\end{eqnarray}
\end{proof}

The next lemma appeared e.g. in \cite{AFH09}.

\begin{lemma} \label{boundongapofromlocal}
\cite{AFH09} Let $H = \sum_{i=1}^N H_{i, i+1}$ be a local Hamiltonian with periodic boundary conditions, with $H_{i, i+1}$ projectors and $\lambda_{\min}(H) = 0$. Then the spectral gap of $H$ satisfies
\begin{equation}
\Delta(H) \geq 2 \min_{i \in {1, ..., N}} \Delta \left( H_{i, i+1} + H_{i+1,i+2} \right) - 1   
\end{equation}
\end{lemma}
\begin{proof}
Let $\gamma := 2\min_{i \in {1, ..., N}} \Delta \left( H_{i, i+1} + H_{i+1,i+2} \right) - 1$. Then by Lemma \ref{gaptrick}
\begin{equation}
\left( H_{i, i+1} + H_{i+1,i+2} \right)^2 \geq \frac{\gamma + 1}{2} \left( H_{i, i+1} + H_{i+1,i+2} \right). 
\end{equation}
Rearranging terms in the equation above we get
\begin{equation}
 \frac{1}{2}H_{i, i+1} +  H_{i, i+1}H_{i+1,i+2} + H_{i+1,i+2}H_{i, i+1} + \frac{1}{2}H_{i+1,i+2} \geq \frac{\gamma}{2} \left( H_{i,i+1} + H_{i+1,i+2} \right).
\end{equation}

We have
\begin{eqnarray} \label{eqHsquaretogamma}
H^2 &\geq& \sum_{i=1}^N \left( \frac{1}{2}H_{i, i+1} +  H_{i, i+1}H_{i+1,i+2} + H_{i+1,i+2}H_{i, i+1} + \frac{1}{2}H_{i+1,i+2} \right) \nonumber \\ &\geq& \gamma \sum_{i=1}^N \left( \frac{1}{2} H_{i, i+1} + \frac{1}{2} H_{i+1, i+2} \right) = \gamma H,
\end{eqnarray}
where the first inequality follows from
\begin{eqnarray}
H^2 &=& \sum_{i=1}^N \left( \frac{1}{2}H_{i, i+1} +  H_{i, i+1}H_{i+1,i+2} + H_{i+1,i+2}H_{i, i+1} + \frac{1}{2}H_{i+1,i+2} \right) \nonumber \\ 
&+& \sum_{ |k-l|>1 }  H_{k, k+1} \otimes H_{l, l+1}, 
\end{eqnarray}
and the positivity of $H_{i, i+1}$, and the second one from Eq. (\ref{eqHsquaretogamma}). Then, by Lemma (\ref{gaptrick}), $\Delta(H) \geq \gamma$.
\end{proof}

\begin{lemma} \label{gaptrick}
For a positive semi-definite matrix $M$ with $\lambda_1(M) = 0$,  
\begin{equation}
\lambda_{2}(M) = \max_{\gamma \in \mathbb{R}} \gamma : M^2 \geq\gamma M
\end{equation}
\end{lemma}
\begin{proof}
Let $(0, \lambda_2, \lambda_3, ...)$ be the vector of ordered eigenvalues of $M$. Then $(0, \lambda_2^2, \lambda_3^2, ...)$ are the ordered eigenvalues of $M^2$ and $M^2 \geq \gamma M$ holds true if, and only if, $\lambda_k^2 \geq \gamma \lambda_k$ for all $k$, which is equivalent to $\lambda_2 \geq \gamma$. 
\end{proof}

The second result we need in the proof of Theorem \ref{3design} is the following upper bound on $\lambda_2(M_{t, 3})$, valid for $t = 2, 3$:
\begin{lemma} \label{gapequals710}
For $t \in \{ 2, 3 \}$,
\begin{equation}
\lambda_2(M_{t, 3}) = \frac{7}{10}.
\end{equation}
\end{lemma}
\begin{proof}

The operator $M_{t,3}$ has the following form
\be
M_{t,3}=\frac12(P_{12}\ot \id_3 +\id_1\ot P_{23})
\ee

In analogy to the definition of $P_{i, i+1}$, we define the projectors $P_i$ as follows
\begin{equation}
P_{i} := \int_{\mathbb{U}(2)} \mu_H(dU)  U_{i}^{\otimes t} \otimes \overline{U}_{i}^{\otimes t}.
\end{equation}
We also consider the associate superoperator 
\begin{equation}
{\cal P}_i(X) := \int_{\mathbb{U}(2)} \mu_H(dU)  U_{i}^{\otimes t} X (U_{i}^{\cal y})^{\otimes t}.
\end{equation}

From Schur duality we find that all operators $X$ invariant under ${\cal P}_{i}$ can be written as a
sum of permutation operators. In more detail, consider the representation $V_{\pi, i}$ of the symmetric group $S_t$ acting 
on ${\cal H}_i^{\otimes t}$ given by 
\begin{equation}
V_{\pi, i} \ket{k_1} \otimes ... \otimes \ket{k_t} = \ket{k_{\pi^{-1}(1)}} \otimes ... \otimes \ket{k_{\pi^{-1}(t)}},
\end{equation}
for every $\pi \in S_t$ and $\ket{k_l} \in {\cal H}_i$. Then it follows that any $X$ satisfying ${\cal P}_i(X) = X$ is such that 
\begin{equation}
X = \sum_{\pi} c_{\pi} V_{\pi, i},
\end{equation}
for complex numbers $c_{\pi}$. Moreover, the subspace defined by $P_i$ is spanned by the (overcomplete) basis given by  the non-normalized vectors $\ket{V_{\pi, i}} := V_{\pi, i} \otimes \id \ket{\Phi}$.

Likewise we define
\begin{equation}
{\cal P}_{i, i+1}(X) := \int_{\mathbb{U}(4)} \mu_H(dU)  U_{i, i+1}^{\otimes t} X (U_{i, i+1}^{\cal y})^{\otimes t},
\end{equation}
and again we find that any $X$ invariant under ${\cal P}_{i, i+1}$ can be written as a sum of permutartion operators $V_{\pi, (i, i+1)}$, now permuting the $t$ copies of the Hilbert space ${\cal H}_{i, i+1}^{\otimes t}$. 

Since $V_{\pi, (i, i+1)} = V_{\pi, i} \otimes V_{\pi, i+1}$ it follows that $P_{i,i+1}\subset P_i\ot P_{i+1}$ and thus
\begin{equation}
P_{12}\ot \id_3 +\id_1\ot P_{23}=
(P_{12}\ot P_3 +P_1\ot P_{23}) \oplus (P_{12}\ot P^\perp_3 )\oplus (P^\perp_1\ot P_{23})
\end{equation}
Since the last two terms in the direct sum above have eigenvalues $0$ and $1$, it follows that $\lambda_2(M_{t,3})=\lambda_2(X)/2$, with $X$ given by 
\be
X=P_{12}\ot P_3 +P_1\ot P_{23}.
\ee 
Note that the support of $X$ is contained in $P_1\ot P_2\ot P_3$. 


Now, for any Hilbert space $\hcal$, in a subspace spanned by swaps acting on 
$\hcal^{\ot t}$ one can construct a basis of operators which is 
orthonormal in Hilbert-Schmidt scalar product:
\be
R_k=\sum_{\pi} b_{k\pi} V_{\pi}
\label{eq:ort-basis}
\ee
where the coefficients $b_{k\pi}$ do not depend on the dimension of $\hcal$,
but only on $t$ (we do not write this dependence explicitly). 
Note that depending on the dimension of $\hcal$, some $R_k$ may vanish. 
Using the operators $R_k$ we can write:
\be
P_{i,i+1}=\sum_k|R_k^{(i,i+1)}\>\<R_k^{(i,i+1)}|, \quad P_i=
\sum_k|R_k^{(i)}\>\<R_k^{(i)}|.
\ee
To diagonalize the operator $X$ we need to represent $P_{1,2}$ and $P_{2, 3}$ 
in terms of product basis $R^{(1)}_k\ot R^{(2)}_k$ and $R^{(2)}_k\ot R^{(3)}_k$, 
respectively. To this end we use \eqref{eq:ort-basis}, which in our case read
\bea
&&R_k^{(1,2)}=\sum_\pi b_{k\pi} V_{\pi, 1}\ot V_{\pi, 2}\nonumber\\
&&R_k^{(1)}=\sum_\pi b_{k\pi} V_{\pi, 1}\nonumber\\
&&R_k^{(2)}=\sum_\pi b_{k\pi} V_{\pi, 2}.\\
\eea
A simple calculation gives
\be
R_k^{1,2}=\sum_{s,u} r^{(k)}_{s,u} R_s\ot R_u
\ee
where the coefficients $r^{(k)}_{s,u}$ form a matrix given by
\be
r^{(k)}=(B^{-1})^T A^{(k)} B^{-1}
\ee
with $B$ defined as the matrix with entries $b_{kl}$ and $A^{(k)}$ the diagonal matrices
\be
A^{(k)}_{ij}=\delta_{ij} b_{ki}.
\ee

In this way, from the matrix $B$ we can obtain the matrix elements 
of the projectors $P_{1,2}$ and $P_{2, 3}$ in the product bases $R^{(1)}_k\ot R^{(2)}_l$ 
and $R^{(2)}_k\ot R^{(3)}_l$, respectively. These, in turn, allow us to obtain the matrix elements of $X$ and to compute 
its eigenvalues. In the following we perform the analysis for $t=2$ and $t=3$. 

{\it 2-design.} For $t=2$, the basis is given 
by (suitably normalized) projectors onto the symmetric and the antisymmetric subspace.
The matrix $B$ is then given by 
\be
B=
\left[
\begin{array}{cc}
\frac{1}{\sqrt3}&\frac{1}{\sqrt3}\\
\frac{1}{2}&-\frac{1}{2}\\
\end{array}
\right]
\ee
where the basis of swaps was ordered as $\{(), (12)\}$ \footnote{Here $()$ labels the trivial permutation and $(12)$ the swap of systems 1 and 2.}, while the matrices $r^{(k)}$ are given by 
\be
r^{(1)}=
\frac12\left[
\begin{array}{cccc}
0&0&0&0\\
0&1&1&0\\
0&1&1&0\\
0&0&0&0\\
\end{array}
\right],\quad
r^{(2)}=
\frac12\left[
\begin{array}{cccc}
\alpha&0&0&\sqrt{\alpha\beta}\\
0&0&0&0\\
0&0&0&0\\
\sqrt{\alpha\beta}&0&0&\beta\\
\end{array}
\right].
\ee
with $\alpha=\frac95,\beta=\frac15$. Diagonalizing $X$, we obtain that it has the second largest eigenvalue $7/5$ \footnote{We have diagonalized the matrix $X$ using the symbolic manipulation program \textit{Mathematica}.},
which gives $\lambda_2(M_{3,t}) = 7/10$ for $t=2$.

{\it 3-design.}
Here we exploit the orthonormal basis constructed in \cite{UUU}:
\bea
&&R_+=\frac{1}{12} (\id + V_{(12)}+V_{(13)}+V_{(23)} + V_{(123)}+V_{(132)}),\nonumber\\
&&R_-=\frac{1}{12} (\id -V_{(12)}-V_{(13)}-V_{(23)} + V_{(123)}+V_{(132)}),\nonumber\\
&&R_0=\frac12(\id-R_+-R_-),\nonumber\\
&&R_1=\frac16(2V_{(23)}-V_{(13)}- V_{(12)}),\nonumber\\
&&R_2=\frac{1}{2\sqrt{3}}(V_{(12)}-V_{(13)}), \nonumber\\
&&R_3=\frac{i}{2\sqrt{3}}(V_{(123)}-V_{(132)}).
\eea
where permutations are written by means of cycles. The related matrix $B$ is given by 
\be
B=\frac12\left[
\begin{array}{cccccc}
\frac23&0&0&0&-\frac13&-\frac13\\
0&-\frac13&\frac23&-\frac13&0&0\\
0&\frac{1}{\sqrt{3}}&0&-\frac{1}{\sqrt{3}}&0&0\\
0&0&0&0&\frac{i}{\sqrt{3}}&-\frac{i}{\sqrt{3}}\\
\frac16&\frac16&\frac16&\frac16&\frac16&\frac16\\
\frac16&-\frac16&-\frac16&-\frac16&\frac16&\frac16\\	
\end{array}
\right]
\ee
where the basis of swaps is ordered as follows $\{(),(12),(23),(13),(123),(132)\}$.
Since we work with qubits, $R_3=0$ and hence we have $5$ basis elements and $X$ is 
a $125 \times 125 $ matrix. We have computed its 
matrix elements and then its eigenvalues$^{11}$, and found that again the second largest eigenvalue is equal to $7/5$.  
\end{proof}

\vspace{0.2 cm}
Lastly, let us discuss the case of uniform random circuits. We have: 
\begin{proposition}
$5n^{2}\log(1/\epsilon)$-sized uniform random circuits form an $\epsilon$-approximate unitary 3-design.
\end{proposition}
\begin{proof}
Following the proof of Theorem \ref{3design} is suffices to show that for $t=3$,
\begin{equation} \label{boundgapuniformmodel}
\lambda_{2}(N_{t, n}) \leq 1 - \frac{1}{5n^{2}}
\end{equation}
with
\begin{equation}
N_{t, n} := \frac{1}{n^{2}} \sum_{i < j} P_{i, j}
\end{equation}
and
\begin{equation}
P_{i, j} := \int_{\mathbb{U}(4)} \mu_{H}(dU) U_{i, j}^{\otimes t} \otimes \overline{U}_{i, j}^{\otimes t}.
\end{equation}

Defining
\begin{equation}
\tilde H := \sum_{i, j} H_{i, j} = n^{2}(\id - N_{t, n})
\end{equation}
with $H_{i, j} = \id - P_{i, j}$, we find Eq. (\ref{boundgapuniformmodel}) to be equivalent to $\Delta(\tilde H) \geq 1/5$. 

Both $\tilde H$ and $H$ (given by Eq. (\ref{defH})) have the same groundspace $S_{0}$. Furtermore, $\tilde H \geq H$. Then
\begin{equation}
\Delta(\tilde H) \geq \min_{\ket{\psi} \perp S_{0}} \bra{\psi} \tilde H \ket{\psi} \geq  \min_{\ket{\psi} \perp S_{0}} \bra{\psi} H \ket{\psi} = \Delta(H) \geq \frac{1}{5},
\end{equation}
where the last inequality follows from Lemmas \ref{gaponnqubitsfromgapon3qubits} and \ref{gapequals710}.
\end{proof}

\section{Acknowledgments}

We would like to thank Salman Beigi for useful correspondence. FB is supported by a "Conhecimento Novo" fellowship from the Brazilian agency Funda\c{c}\~ao de Amparo a Pesquisa do Estado de Minas Gerais (FAPEMIG). MH is supported by EU grant QESSENCE and by Polish Ministry of Science and Higher
Education grant N N202 231937. FB would like to thank the hospitality of the National Quantum Information Center of Gdansk where part of this work was done. FB and MH thank the Intitute Mittag Leffler for their hospitality, where (another) part of this work was done.

\end{document}